\documentclass[letterpaper, 10 pt, conference]{ieeeconf}  

\IEEEoverridecommandlockouts                              
\usepackage{amssymb}
\usepackage{mathrsfs}
\usepackage{algorithm}
\usepackage{algpseudocode}
\usepackage[noadjust]{cite}
\usepackage{bm}
\usepackage{amsmath}
\usepackage{commath}
\usepackage[utf8]{inputenc}
\usepackage{enumerate}
\usepackage[hang,flushmargin]{footmisc}
\usepackage{tikz}
\usetikzlibrary{arrows}
\usetikzlibrary{shapes.geometric, shapes.multipart}
\tikzstyle{block}=[rectangle,text centered]
\usepackage{mathtools}

\usepackage{amsthm}
\newtheorem{thrm}{Theorem}

\newtheorem{defn}{Definition}
\newtheorem{remark}{Remark}

\newtheorem{assumpt}{Assumption}

\theoremstyle{definition}
\title{\LARGE \bf
Adaptive\! Model\! Predictive\! Safety\! Certification\! for\! Learning-based\! Control}

\author{Alexandre Didier, Kim P. Wabersich and Melanie N. Zeilinger 
}

\begin{document}
\maketitle
\thispagestyle{empty}
\pagestyle{empty}

\begin{abstract}
We propose an adaptive Model Predictive Safety Certification (MPSC) scheme for learning-based control of linear systems with bounded disturbances and uncertain parameters with known bounds. An MPSC is a modular framework, which can be used in combination with any learning-based controller to ensure state and input constraint satisfaction of a dynamical system by solving an online optimisation problem. By continuously connecting the current system state with a safe terminal set using a robust tube, safety can be ensured. Thereby, the main sources of conservative safety interventions are model uncertainties and short planning horizons. We develop an adaptive mechanism to improve the system model, which leverages set-membership estimation to guarantee recursively feasible and non-decreasing safety performance improvements. In order to accommodate short prediction horizons, iterative safe set enlargements using previously computed robust backup plans are proposed. Finally, we illustrate the increase of the safety performance through the parameter and safe set adaptation for numerical examples with up to 16 state dimensions. 

\end{abstract}

\section{INTRODUCTION}\label{sec:introduction}
Learning-based control is seeing growing interest due to the abundance of data being collected in today's control systems. Especially reinforcement learning has demonstrated that controllers can be learned for complex or even uncertain cost functions and system models, see e.g. \cite{ng2006autonomous}, \cite{vanhasselt2016deep}. However, these methods often lack safety guarantees, i.e., the proposed control actions of the learning-based algorithm could lead the system into unsafe regions of the state space, e.g., a quadrotor approaching the ground with high speed, especially during exploration. This limits their application to safety-critical systems, e.g., autonomous transportation systems or medical applications, where certain state and input constraints are required to be satisfied for safety. 

In order to leverage the advantages from learning-based control while ensuring constraint satisfaction, modular, invariance-based safety frameworks have been developed using control barrier functions, see e.g. \cite{wieland2007constructive} and \cite{ames2019control}, or Hamilton-Jacobi reachability, as discussed, e.g., in \cite{gillulay2011guaranteed} and \cite{chen2018hamilton}.
As these approaches can be computationally challenging or difficult to design in the case of larger scale systems, they have been extended using Model Predictive Control (MPC) techniques, see e.g. \cite{rawlings2017model}, providing a scalable safety framework for linear dynamics in \cite{wabersich2018linear}, with extensions for probabilistic, nonlinear or distributed systems in \cite{wabersich2021probabilistic}, \cite{wabersich2021safe} and \cite{muntwiler2020distributed}, respectively. Here, a predictive control problem is solved at every time step to find the closest input to a proposed learning-based input together with a trajectory satisfying all state and input constraints and leading to a terminal safe set. This safe set is a set in the state space, which ensures that constraint satisfaction can be guaranteed at all future time steps through the use of a safety controller. The approach itself implicitly defines a safe set through the feasible set of the predictive control problem, ensuring the existence of a safe backup trajectory for the system. 
While existing formulations \cite{wabersich2021probabilistic,wabersich2021safe,muntwiler2020distributed} are tailored to specific model classes, they do not provide a principled mechanism to adaptively refine the underlying system model using incoming state measurements, while maintaining recursive feasibility guarantees.

In this paper, we propose an adaptive Model Predictive Safety Certification (MPSC) scheme, which considers linear models with parametric uncertainties and unknown but bounded additive disturbances. The proposed scheme allows to augment any learning-based controller such that state and input constraint satisfaction properties are ensured for all future time steps. 
Instead of performing episodic model learning updates, we leverage recent results from adaptive MPC literature, see \cite{lorenzen2019robust}, \cite{lu2019robust}, \cite{kohler2019linear} and \cite{kohler2020nonlinear}, to estimate uncertain parameters in the system dynamics online. By using set-membership estimation, implausible model parametrisations are recursively eliminated, see \cite{milanese1991optimal}. 
This results in a rigorous adaptive refinement of the MPSC scheme, which ensures safety with respect to the uncertain parameters as well as exogenous disturbances through recursive feasibility and guarantees a non-deteriorating performance. 
Through less restrictive assumptions on the terminal set used in the predictive control problem, the design procedure for the proposed adaptive MPSC scheme is simplified compared to previous robust adaptive MPC schemes.
Additionally, we propose a terminal safe set enlargement similar to \cite{bujarbaruah2018adaptive}, which reduces the effect of potentially short planning horizons on the performance due to real-time computation requirements. The terminal safe set enlargement can be performed online using solved instances of the MPSC optimisation problem by using the convex hull of the corresponding tubes.

In this paper we focus on a linear system model as specified in Section~\ref{sec:preliminaries}, for which we derive the proposed method in Section~\ref{sec:ampsc}. 
A discussion of an efficient design procedure for the adaptive MPSC using polytopic disturbance and parameter sets and homothetic tubes is provided in Section~\ref{sec:effampsc}. Finally, a numerical example for a chain of mass-spring-damper systems is provided in Section~\ref{sec:numericalexample} to illustrate the increase in size of the resulting safe sets and a comparison to the MPSC in \cite{wabersich2018linear} is provided.

\section{PRELIMINARIES}\label{sec:preliminaries}
\textit{Notation:} 
The set of integers ranging from $a$ to $b$ is denoted by $\mathcal{I}_{[a,b]}$, the set of all positive integers is $\mathcal{I}_{\geq 0}$ and $2^{\mathbb{A}}$ denotes the power set of the set $\mathbb{A}$.
We define the unit hypercube as $\mathbb{B}_n=\{x\in\mathbb{R}^n|\;\Vert x \Vert_\infty\leq 0.5\}$. The Minkowski sum of two sets $\mathbb{A}\subseteq\mathbb{R}^n$ and $\mathbb{B}\subseteq\mathbb{R}^n$ is given by $\mathbb{A}\oplus\mathbb{B}=\{a+b|\;a\in\mathbb{A}, b\in\mathbb{B}\}$ with $a,b\in\mathbb{R}^n$. 
The convex hull of a set $\mathbb{A}$ is denoted as $\textup{co}(\mathbb{A})$ and the $i$-th entry of the vector $a$ is denoted $[a]_i$. The projection of a set $\mathbb{A}\subset\mathbb{R}^m$ onto the first $n$ dimensions, where $m\geq n$, is given by $\textup{Proj}^n(\mathbb{A})$ and onto the last $n$ dimensions by $\textup{Proj}_n(\mathbb{A})$.
\subsection{Problem Description}\label{sec:problemdescription}
We consider uncertain discrete-time linear dynamics of the form
\begin{equation} \label{eq:lindynamics}
	x_{k+1}=A(\theta)x_k+B(\theta)u_k+w_k,
\end{equation}
with states $x_k\in\mathbb{R}^n$, inputs $u_k\in\mathbb{R}^m$, disturbances $w_k\in\mathbb{W}\subseteq\mathbb{R}^n$ and uncertain parameters $\theta\in\mathbb{R}^p$.
We assume that the true system dynamics are captured by \eqref{eq:lindynamics} with parameters equal to their true value $\theta=\theta^*$. 
The considered disturbance is bounded in a compact set $\mathbb{W}$ and the parameters $\theta$ lie within an a priori known, compact set of parameters $\Theta_0$, which includes the true value $\theta^*$. 
\begin{remark}
The considered problem description also captures nonlinear systems, where a range of parameters $\theta$ can explain the system evolution if the disturbance set $\mathbb{W}$ is enlarged to encompass the error between the considered linear model and the true nonlinear dynamics. An extension to fully support nonlinear dynamics models is given in Appendix~\ref{sec:nonlinearextension}.
\end{remark}
This is a common problem setup used in robust adaptive model predictive control frameworks, see e.g. \cite{lorenzen2019robust}, \cite{lu2019robust} and \cite{kohler2019linear}. The uncertain parameters $\theta$ are assumed to enter the dynamics \eqref{eq:lindynamics} affinely as follows:

\begin{assumpt}\label{ass:bounded}
The system matrices $A(\theta)$ and $B(\theta)$ depend affinely on the parameter vector $\theta\in\mathbb{R}^p$ such that
\begin{equation}
(A(\theta),B(\theta))=(A_0,B_0)+\sum_{i=1}^p(A_i,B_i)[\theta]_i,
\end{equation}
where $A_0, A_i\in\mathbb{R}^{n\times n}$ and $B_0,B_i\in\mathbb{R}^{n\times m}$
\end{assumpt}
Note that such a model description can be derived from a linear system model \eqref{eq:lindynamics} by reformulating parameters which affect the system matrices nonlinearly as new parameters $\theta$ if their influence can be bounded, as is done, e.g., in \cite{didier2021robust}.
The system \eqref{eq:lindynamics} is subject to polytopic safety-critical state and physical input constraints given by
\begin{equation}\label{eq:stateinputconstr}
(x_k,u_k)\in\mathbb{Z}=\{(x,u)\in\mathbb{R}^n\times\mathbb{R}^m|\;Fx+Gu\leq z\},
\end{equation}
where $F\in\mathbb{R}^{n_z\times n}$, $G\in\mathbb{R}^{n_z\times m}$ and $z\in\mathbb{R}^{n_z}$. 
The projection of the constraint set $\mathbb{Z}$ onto the state space $\mathbb{R}^n$ and input space $\mathbb{R}^m$ is defined as 
$\mathbb{X}=\textup{Proj}^n(\mathbb{Z})$ and 
$\mathbb{U}=\textup{Proj}_m(\mathbb{Z})$, respectively.


\subsection{Parameter Identification}\label{sec:parameteridentification}
Instead of inferring parameter estimates a priori from data as done in, e.g., \cite{wabersich2021safe}, we begin with a set of possible parameters, which will iteratively be refined online using incoming state measurements. More precisely, starting from an initial uncertainty set $\Theta_0$, which could arise in practice from, e.g., production tolerances or tasks with uncertain parameters like lifting an object with uncertain mass as in \cite{didier2021robust}, new sets $\Theta_k$ are inferred with the properties given in the following assumption.
\begin{assumpt}\label{ass:parameter}
The parameter identification method fulfils for all $k\geq0$
\begin{enumerate}
\item Consistency of the identification method, i.e., if the true parameter $\theta^*\in\Theta_0\Rightarrow\theta^*\in\Theta_k$
\item Recursive set estimate inclusion, i.e., $\Theta_{k+1}\subseteq\Theta_{k}\subseteq\Theta_0$
\end{enumerate}
\end{assumpt}
Note that this assumption encompasses any parameter identification method for which a set of parameters is guaranteed to contain the true parameter. If consecutive sets are not recursively contained within each other, e.g., due to restrictions on $\Theta_k$ for computational reasons, the sets can be updated only when they are a subset of the previously used set. For example, if confidence sets obtained via Bayesian Linear Refression are used for the parameters such that $\textup{Pr}(\theta^*\in\Theta_k)\geq p_\theta$ for some desired probability level $p_\theta$, then recursive inclusion of the set estimates is not guaranteed given new data and needs to be verified online.
Different set-membership estimation methods exist that fulfil the properties in Assumption~\ref{ass:parameter} by construction, such as a polytopic formulation in \cite{milanese1991optimal} and a spherical formulation in \cite{dhaliwal2012set}. By using such an adaptive model refinement, we derive the adaptive MPSC scheme in the following section. The computation of polytopic parameter sets using set-membership estimation is detailed in Section~\ref{sec:efficientAMPSC}, which allows for a computationally efficient adaptive MPSC scheme.



\section{ADAPTIVE MODEL PREDICTIVE SAFETY CERTIFICATION}\label{sec:ampsc}
The proposed adaptive MPSC scheme is a modular framework, which takes as an input a learning-based control action $u_k^\mathscr{L}$ and the current state in order to verify the safety of the proposed action based on computing a safe forward plan using a sequentially improved data-driven model. 
A schematic of this framework can be seen in Figure \ref{fig:schematic}, where the applied control input corresponds to the MPSC policy, i.e., $u_k=\pi_{\textup{MPSC}}(u_k^\mathscr{L},x_k,\Theta_k,k)$.

\begin{figure}[h]
\vspace{0.2cm}
   \centering
   \begin{tikzpicture}[node distance=2cm, every text node part/.style={align=center}]

	\node (System) [rectangle, text centered, draw=black, minimum height=1.3cm, minimum width=4cm] {System \\ $x_{k+1}=A(\theta)x_k+B(\theta)u_k+w_k$};

	\node (Learning) [rectangle, text centered, draw=black, minimum height=1.3cm, minimum width= 3.5cm,below of=System, xshift=2.2cm, yshift=0.2cm] {Learning-based \\ Controller};

	\node (Parameter) [rectangle, text centered, draw=black, minimum height=1.3cm, minimum width= 3.5cm,below of=Learning, yshift=0.2cm] {Parameter Estimation};

	\node (Adaptive) [rectangle, text centered, draw=black, minimum height=1.3cm, minimum width= 3.5cm,below of=System, xshift=-2.2cm, yshift=-0.7cm] {Adaptive MPSC};

	\draw [->, to path={-| (\tikztotarget)}] (System)  --  node[anchor=south] {$x_k$} (4.2cm,0) |- (Learning.east);

	\draw [->, to path={-| (\tikztotarget)}] (System)  --  (4.2cm,0) |- (Parameter.east);

	\draw [->, to path={-| (\tikztotarget)}] (System)  --  (4.2cm,0) |- (Adaptive.east);

	\draw [->, to path={-| (\tikztotarget)}] (Learning.west)  --  node[anchor=south] {$u_k^\mathscr{L}$} (-0.2cm, -1.8cm) |- ([yshift=0.3cm]Adaptive.east);

	\draw [->, to path={-| (\tikztotarget)}] (Parameter.west)  --  node[anchor=north] {$\Theta_k$} (-0.2cm, -3.6cm) |- ([yshift=-0.3cm]Adaptive.east);

	\draw [->, to path={-| (\tikztotarget)}] (Adaptive.west)  -|  (-4.2cm, 0cm) -- node[anchor=south] {$u_k$} (System.west);

   \end{tikzpicture}
   \caption{Schematic of the adaptive Model Predictive Safety Certification framework. Given the current state of the system $x_k$, a learning-based control input $u_k^\mathscr{L}$ and the set of possible parameters $\Theta_k$, resulting from the parameter estimation, 
			the input $u_k$ to be applied to the system is provided by the adaptive Model Predictive Safety Certification scheme.}
   \label{fig:schematic}
\end{figure}
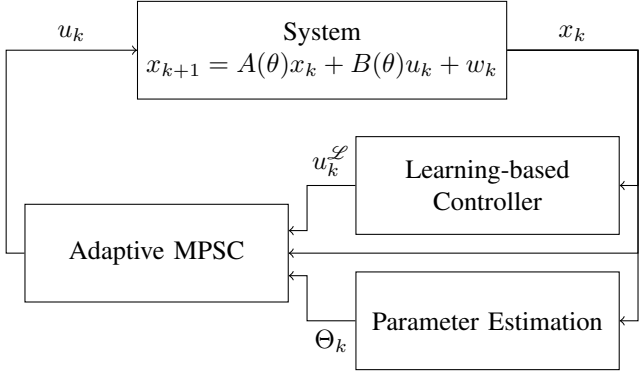
The proposed method is based on computing a state and input backup trajectory from the current state to a terminal safe set, with the goal of matching the first element of the input backup sequence with the desired learning input at each time step.
If the final predicted backup states are contained in the terminal safe set, constraint satisfaction can be guaranteed at all further time steps.
In the following, we begin by formalising the terminal safe set, which is used to define the MPSC algorithm in Section~\ref{sec:ampscalg}. The optimisation problem, which is solved at every time step in order to compute the backup trajectory is then discussed and the algorithm of the adaptive MPSC is provided. Finally, extensions of the scheme are presented by updating the terminal safe set. 

\subsection{Adaptive Model Predictive Safety Certification Algorithm} \label{sec:ampscalg}
In order to guarantee that the constraints \eqref{eq:stateinputconstr} can be satisfied for all times, the concept of a safe set is used, as defined in \cite{wabersich2018linear}, \cite{wabersich2021probabilistic}, \cite{wabersich2021safe} and \cite{wabersich2018scalable}. A safe set is a set in the state space, which ensures constraint satisfaction through the use of a safe control policy $\pi_\mathbb{S}$. 

\begin{defn}\label{def:safeset}
	A set $\mathbb{S}\subseteq\mathbb{X}$ is called a safe set for system \eqref{eq:lindynamics} if a safe backup control law $\pi_\textup{B}:\mathbb{R}^m\times\mathbb{R}^n\times 2^{\mathbb{R}^p}\times\mathcal{I}_{\geq0}\rightarrow\mathbb{U}$ is available such that for an arbitrary (learning-based) action $u_k^\mathscr{L}\in\mathbb{R}^m$, the application of the safe control policy
\begin{gather}
	\begin{align}
	&\pi_{\mathbb{S}}(u_k^\mathscr{L}, x_k, \Theta_k, k)=\nonumber \\
	&\begin{cases}
		u_k^\mathscr{L}, \textup{ if }(x_k, u_k^\mathscr{L}){\in}\mathbb{Z} {\wedge}\{A(\theta)x_k{+}B(\theta)u_k^\mathscr{L}\}{\oplus}\mathbb{W}{\subseteq}\mathbb{S}\;\forall\theta{\in}\Theta_k \\
		\pi_\textup{B}B(u_k^\mathscr{L}, x_k, \Theta_k, k), \textup{ otherwise}
	\end{cases}\nonumber
	\end{align}
\end{gather}
guarantees constraint satisfaction of the system state and inputs, i.e., $(x_k, \pi_{\mathbb{S}}(u_k^\mathscr{L}, x_k, \Theta_k, k))\in\mathbb{Z}$ for all $k\geq\bar{k}$ if $x_{\bar{k}}\in\mathbb{S}$.
\end{defn}
A safe set thus provides a guarantee that the state and input constraints are satisfied for all times $k\geq\bar{k}$ by using the safe control policy $\pi_{\mathbb{S}}(u_k^\mathscr{L}, x_k, \Theta_k, k)$ if the state $x_{\bar{k}}$ is in the safe set $\mathbb{S}$ at time step $\bar{k}$. 
Note that for the convex polytopes resulting from set-membership estimation, it suffices to check the condition $\{A(\theta^j_k)x_k+B(\theta^j_k)u_k^\mathscr{L}\}\oplus\mathbb{W}\subseteq\mathbb{S}$ at every vertex $\theta_k^j$ of the polytope $\Theta_k$.

\begin{remark}
While a Robust Positively Invariant (RPI) set used in robust MPC, see e.g. \cite[Chapter 2.6]{rawlings2017model}, requires that any possible state evolution starting inside the set will be contained in the set, Definition~\ref{def:safeset} of a safe set only requires that starting from a certain subset implies safety for all future times. This allows, e.g., to define a safe set using expert knowledge without the need for expensive offline computations.
However, principled robust invariant set computations can be employed and are available for parametric uncertainties, e.g., according to the algorithm provided in \cite{pluymers2005efficient}, even though they suffer from limited scalability. Ellipsoidal RPI sets for linear feedback controllers can be computed through semi-definite programming, see e.g. \cite{wabersich2018scalable} and \cite[Appendix]{kohler2019linear}.
\end{remark}

Due to the uncertain model, the computation of the required backup trajectory can be conservative, motivating the use of recent advances in robust adaptive MPC schemes \cite{lorenzen2019robust}, \cite{lu2019robust}, \cite{kohler2019linear} and \cite{kohler2020nonlinear}. At every time step $k$, we compute a tube in the state space starting from the current state measurement $x_k$, which is guaranteed to contain the future states for any disturbances in $\mathbb{W}$ and uncertain parameters in $\Theta_k$ through the use of a tube control law $\kappa:\mathbb{R}^n\times\mathbb{R}^{n_v}\rightarrow\mathbb{R}^m$ with $n_v$ parameters. The tube then consists of sets $\mathbb{X}_{l|k}$, which are predicted at time $k$, given the tube control law $\kappa(x,v_{l|k})$, for $l\in\mathcal{I}_{[0,N]}$ future time steps given a horizon $N$, with the last polytope being constrained to lie in a terminal safe set $\mathbb{S}_f$. The last predicted set $\mathbb{X}_{N|k}$ is required to be a subset of a terminal safe set $\mathbb{S}_f$, which fulfils Definition~\ref{def:safeset} with the safe control policy $\pi_{\mathbb{S}_f}(u_k^\mathscr{L}, x_k, \Theta_k, k)$. This allows to ensure constraint satisfaction for further time steps if 
for all $x\in\mathbb{X}_{l|k}$, it holds that $(x,\kappa(x,v_{l|k}))\in\mathbb{Z}$.

The adaptive MPSC algorithm is a modular framework which uses a learning-based controller for performance, i.e., the goal is to apply the learning-based input if constraint satisfaction can be ensured. This objective is realised by minimising the norm of the difference of the first control input of the planned tube and the proposed learning-based input $u_k^\mathscr{L}$. The optimisation problem we solve at every time step is thus given as
\begin{subequations} \label{eq:optimisation}
\begin{alignat}{1}
\min_{v_{\cdot |k}, \mathbb{X}_{\cdot|k}} & \Vert u_k^\mathscr{L}-\kappa(x_k, v_{0|k})\Vert \\
\textup{s.t. } & \forall l\in\mathcal{I}_{[0,N-1]} \nonumber \\
& x_k \in\mathbb{X}_{0|k},  \label{eq:initialconstr} \\
& A(\theta)x+B(\theta)\kappa(x,v_{l|k})+w \in \mathbb{X}_{l+1|k}, \nonumber\\
&\qquad\qquad \forall x\in\mathbb{X}_{l|k}, w\in\mathbb{W}, \theta\in\Theta_k,\label{eq:tubeincconstr} \\
&(x,\kappa(x,v_{l|k}))\in\mathbb{Z}, \quad \forall x \in\mathbb{X}_{l|k}, \label{eq:stateconstr} \\
&\mathbb{X}_{N|k}\subseteq \mathbb{S}_f. \label{eq:terminalconstr}
\end{alignat}
\end{subequations}

As \eqref{eq:optimisation} is not guaranteed to be recursively feasible due to the weak terminal safe set assumption compared to a robust invariant set, a switching mechanism is introduced similar to \cite{wabersich2018linear} in case the optimisation problem becomes infeasible. The mechanism then switches to the last computed optimal solution of \eqref{eq:optimisation} at time step $\bar{k}$, i.e., $\kappa(x_{\bar{k}+l}, v^*_{l|\bar{k}})$ for $l\in\mathcal{I}_{[1,N]}$. This input sequence guarantees that the state reaches the terminal safe set according to \eqref{eq:terminalconstr}. At this point, if \eqref{eq:optimisation} remains infeasible, the backup controller according to Definition~\ref{def:safeset} is used, such that safety is ensured for all time steps. The described procedure is formalised in Algorithm~1. 
\begin{remark}
For a less intrusive safety filter algorithm, Line~10 in Algorithm~1 can be replaced with \\
\centerline{\textup{10: Solve \eqref{eq:optimisation} with horizon $N-k_{\textup{inf}}$,}}\\
which preserves the safety guarantees, similar to \cite{wabersich2021safe}.
\end{remark}

If the initial state $x_0$ of the system lies within the feasible set of \eqref{eq:optimisation} for the initial unknown parameter set $\Theta_0$, denoted as $\mathbb{X}_{\textup{feas}}(\Theta_0)$, or within the terminal safe set $\mathbb{S}_f$, Algorithm~1 guarantees constraint satisfaction for all time steps $k\geq0$ by construction. This follows from the set update of $\Theta_k$ in Line 3 of Algorithm~1, which ensures that $\theta^*\in\Theta_k$ for all $k$ if $\theta^*\in\Theta_0$ under Assumption~\ref{ass:parameter}. It is thus possible to show that the feasible set of \eqref{eq:optimisation} implicitly describes a safe set. Additionally, through the update of the parameter set $\Theta_k$ under Assumption~\ref{ass:parameter}, the size of the feasible set increases as the parameter estimate improves, i.e., $\mathbb{X}_{\textup{feas}}(\Theta_{k-1})\subseteq\mathbb{X}_{\textup{feas}}(\Theta_k)$ as $\Theta_{k-1}\supseteq\Theta_k$.
\vspace{-0.2cm}
\begin{algorithm}[H]\label{alg:1}
\caption{Adaptive Model Predictive Safety Certification Scheme.}
\begin{algorithmic}[1]
\State $k_{\textup{inf}}\leftarrow N-1$
\For{$k=0,1,\dots$} 
	\State Update $\Theta_k$ using the state measurement $x_k$
	\If{\eqref{eq:optimisation} is feasible}
		\State Apply $u_k\leftarrow \kappa(x_k,v^*_{0|k})$ to \eqref{eq:lindynamics}
		\State $k_{\textup{inf}}\leftarrow0$
	\Else
		\State $k_{\textup{inf}}\leftarrow k_{\textup{inf}}+1$
		\If{$k_{\textup{inf}}\leq N-1$}
			\State Apply $u_k\leftarrow \kappa(x_k,v^*_{k_{\textup{inf}}|k-k_{\textup{inf}}})$ to \eqref{eq:lindynamics}
		\Else
			\State Apply $u_k\leftarrow \pi_{\mathbb{S}_f}(u_k^\mathscr{L}, x_k, \Theta_k, k)$ to \eqref{eq:lindynamics}
		\EndIf
	\EndIf
\EndFor
\end{algorithmic}
\end{algorithm}
\vspace{-0.5cm}
\begin{thrm}\label{thrm:AMPSC}
	If Assumptions \ref{ass:bounded} and \ref{ass:parameter} hold, the control law $\pi_{\textup{MPSC}}(u_k^\mathscr{L},x_k,\Theta_k,k)$ resulting from Algorithm~1 is a safe backup controller and the set $\mathbb{X}_{\textup{feas}}(\Theta_k)\cup\mathbb{S}_f$ is the corresponding safe set at time step $k$ according to Definition~\ref{def:safeset}. Additionally, it holds that $\mathbb{X}_{\textup{feas}}(\Theta_{0})\subseteq\mathbb{X}_{\textup{feas}}(\Theta_{1})\subseteq\dots\subseteq\mathbb{X}_{\textup{feas}}(\Theta_{k})$ for all time steps $k>0$.
\end{thrm}
\begin{proof}
The first part of this proof is analogous to the proof of \cite[Theorem~III.5]{wabersich2018linear}. Consider $x_0\in\mathbb{S}_f\setminus\mathbb{X}_{\textup{feas}}(\Theta_0)$, through the initialisation of $k_{\textup{inf}}$, $\pi_{\mathbb{S}_f}(u_k^\mathscr{L}, x_k, \Theta_k, k)$ is applied to the system, which according to Definition~\ref{def:safeset} ensures constraint satisfaction for all future time steps. If $x_0\in\mathbb{X}_{\textup{feas}}(\Theta_0)$ and \eqref{eq:optimisation} is feasible for all $k\geq0$, it follows that safety is ensured through the constraints (\ref{eq:initialconstr}-d) as $\theta^*\in\Theta_k$ under Assumption~\ref{ass:parameter}, see e.g. \cite[Theorem~14]{lorenzen2019robust}. If at any given time step $\bar{k}$, \eqref{eq:optimisation} becomes infeasible, the optimal control input $ \kappa(x_{\bar{k}+k_{\textup{inf}}},v^*_{k_{\textup{inf}}|\bar{k}-1})$ from time step $\bar{k}-1$ is used until $x_{\bar{k}-1+N}\in\mathbb{X}_{\bar{k}-1+N|\bar{k}-1}\subseteq\mathbb{S}_f$ according to \eqref{eq:tubeincconstr} and \eqref{eq:terminalconstr}. At this point, Algorithm~1 switches to using the safe control input $\pi_{\mathbb{S}_f}(u_k^\mathscr{L}, x_k, \Theta_k, k)$. Thus constraint satisfaction is guaranteed by \eqref{eq:stateconstr} and the definition of the terminal safe set.
Through the parameter set update it holds that $\Theta_{k-1}\supseteq\Theta_k$, as follows from Assumption~\ref{ass:parameter}. It therefore holds that any state $x\in\mathbb{X}_{\textup{feas}}(\Theta_{k-1})$ must fulfil $x\in\mathbb{X}_{\textup{feas}}(\Theta_k)$ as constraint \eqref{eq:tubeincconstr} is fulfilled for all $\theta\in\Theta_{k}\subseteq\Theta_{k-1}$.
\end{proof}

\subsection{Iterative Enlargement of the Terminal Safe Set}\label{sec:ittermset}
While the terminal safe set can be enlarged using previously solved instances for adaptive MPC with unknown constant offset as is done in \cite{bujarbaruah2018adaptive}, it has not been discussed for adaptive MPC with parametrised system matrices to the best of the authors' knowledge. 
Using a convex formulation \eqref{eq:optimisation}, it is possible to show that the convex hull of all initial polytopes $\mathbb{X}_{0|k}$ can be added to the terminal safe set. The convex hull of the set of time steps, where \eqref{eq:optimisation} was successfully solved, is denoted as $\mathcal{M}(k)=\{i\in\mathcal{I}_{[0,k]}|\; x_i\in\mathbb{X}_{\textup{feas}}(\Theta_i)\}$ and we use
\begin{equation}
	\mathbb{X}^*_{0|\mathcal{M}(k)}=\textup{co}\left(\{\mathbb{X}^*_{0|i}\}_{i\in\mathcal{M}(k)}\right).
\end{equation}
The terminal safe set $\mathbb{S}_f$ can then be enlarged as follows.

\begin{thrm}\label{thrm:safesetenlargement}
	If Assumptions \ref{ass:bounded} and \ref{ass:parameter} hold and \eqref{eq:optimisation} is convex, then the set 
	\begin{equation}
		\mathbb{S}_f^{\mathcal{M}(k)}=\mathbb{X}^*_{0|\mathcal{M}(k)}\cup\mathbb{S}_f
	\end{equation}
	is again a safe set according to Definition~\ref{def:safeset} with a safe backup controller given by Algorithm~1 with terminal safe set $\mathbb{S}_f$.
\end{thrm}
\begin{proof}
As \eqref{eq:optimisation} is assumed to be convex, it follows that for a fixed parameter set $\Theta_k$, the feasible set $\mathbb{X}_{\textup{feas}}(\Theta_k)$ of \eqref{eq:optimisation} is also convex at every time step $k>0$, see \cite{boyd2004convex}. It then follows that $\mathbb{X}^*_{0|\mathcal{M}(k)}\subseteq\mathbb{X}_{\textup{feas}}(\Theta_k)$ as any $x\in\mathbb{X}^*_{0|i}$ admits a feasible solution to \eqref{eq:optimisation} for all $i\in\mathcal{M}(k)$. As it holds that $\mathbb{X}_{\textup{feas}}(\Theta_{k})\subseteq\mathbb{X}_{\textup{feas}}(\Theta_{k+1})$ if the parameter set is updated and that the union of two safe sets is a safe set, the result follows.
\end{proof}
%

\subsection{A Recursively Feasible MPSC Scheme}\label{sec:recursiveMPSC}
While the safe set according to Definition~\ref{def:safeset} supports an easier design, the resulting implementation becomes more complex due to the required switching mechanism. As an alternative, we additionally consider the case of requiring an RPI terminal set, for which we additionally provide a data-driven design using past data in Section~\ref{sec:ittermset}. In order to provide a recursively feasible optimisation problem \eqref{eq:optimisation}, we require that under the control law $\kappa(x,v)$, a $v$ exists such that all possible uncertain state evolutions from the last predicted state polytope $\mathbb{X}^*_{N|k}$ will be robustly contained in the terminal safe set. 

\begin{assumpt} \label{ass:termset} 
Consider a non-empty terminal set $\mathbb{X}_f$ and a tube control law $\kappa(x,v)$ in \eqref{eq:optimisation}. For every set $\mathcal{X}\subseteq\mathbb{X}_f$, there exists a $v$, such that $(x,\kappa(x,v))\in\mathbb{Z}$ for all $x\in\mathcal{X}$ and such that for all $\theta\in\Theta_0$ it holds that
\begin{equation}A(\theta)\mathcal{X}\oplus B(\theta)\kappa(\mathcal{X},v)\oplus\mathbb{W}\subseteq\mathbb{X}_f. \nonumber
\end{equation}
\end{assumpt}

Under Assumption~\ref{ass:termset}, recursive feasibility of \eqref{eq:optimisation} can be shown. Note that this generalised assumption contains specific robust adaptive MPC formulation such as \cite{lorenzen2019robust}, \cite{lu2019robust}, \cite{kohler2019linear} and \cite{kohler2020nonlinear} as special cases.

\begin{thrm}
	Let $\mathbb{S}_f=\mathbb{X}_f$. If Assumptions \ref{ass:bounded}, \ref{ass:parameter} and \ref{ass:termset} hold, then $\kappa(x_k,v^*_{0|k})$
	is a safe backup control law and $\mathbb{X}_{\textup{feas}}(\Theta_k)$ a corresponding safe set according to Definition~\ref{def:safeset}. In addition, the set $\mathbb{X}_{\textup{feas}}(\Theta_k)$ is a robust positively invariant set for a fixed $\Theta_k$.
\end{thrm}
\begin{proof}
	The proof follows standard recursive feasibility arguments similar to, e.g., \cite{lorenzen2019robust}. Consider \eqref{eq:optimisation} feasible at time step $\bar{k}$. The optimal input sequence $\kappa(x_{\bar{k}},v^*_{l|\bar{k}})$ for $l\in\mathcal{I}_{[1,N]}$ ensures that $x_{l-1|\bar{k}+1}\in\mathbb{X}^*_{l|\bar{k}}$ since $x_{0|\bar{k}+1}\in\mathbb{X}^*_{1|\bar{k}}$ and according to Assumption~\ref{ass:parameter}, $\Theta_{\bar{k}+1}\subseteq\Theta_{\bar{k}}$. As $\mathbb{X}^*_{N|\bar{k}}\subseteq\mathbb{X}_f$, we can set $\mathbb{X}^*_{N|\bar{k}+1}=\mathbb{X}_f$ according to Assumption~\ref{ass:termset}, which fulfills the terminal constraint \eqref{eq:terminalconstr} with $\mathbb{S}_f=\mathbb{X}_f$, such that state and input constraints are satisfied.
Robust positive invariance follows directly from recursive feasibility, as $x_k\in\mathbb{X}_{\textup{feas}}(\Theta_k)\Rightarrow x_{k+1}\in\mathbb{X}_{\textup{feas}}(\Theta_{k})$.
\end{proof}

The design of a terminal set $\mathbb{X}_f$ fulfilling Assumption~\ref{ass:termset} for homothetic tube sets $\mathbb{X}_{l|k}$ is discussed in \cite{lorenzen2019robust}, \cite{kohler2019linear} and a low-complexity terminal set for a 12-dimensional quadrotor example is presented in \cite{didier2021robust}, whereas the condition is implemented as a constraint in the optimisation problem in \cite{lu2019robust}. Note that if a terminal set $\mathbb{X}_f$ fulfills Assumption~\ref{ass:termset}, a terminal safe set enlargement similar to Section~\ref{sec:ittermset} can be performed using the convex hull of all computed solutions $\mathbb{X}_{l|k}^*$ and the terminal set $\mathbb{X}_f$, as feasibility of \eqref{eq:optimisation} is guaranteed. The resulting set is then a safe set according to Definition~\ref{def:safeset}, but does not verify Assumption~\ref{ass:termset}, for which we need a different approach tailored to a specific tube structure as presented in Section~\ref{sec:ittermrecfeas}.


\section{EFFICIENT DESIGN USING POLYTOPIC SETS}\label{sec:effampsc}
In this section, we provide details on how a computationally efficient adaptive MPSC problem can be designed for the linear case by leveraging the formulations in \cite{lorenzen2019robust}, \cite{lu2019robust} and \cite{kohler2019linear}. We then show how the specific structure can be exploited to obtain a data-driven terminal set enlargement, resulting in a recursively feasible optimisation problem \eqref{eq:optimisation}. 
\subsection{Homothetic Tube Formulation}\label{sec:efficientAMPSC}
The considered formulation makes use of recent reformulations of the constraints in \eqref{eq:optimisation} into linear constraints with respect to the optimisation variables in \cite{lorenzen2019robust}, \cite{lu2019robust} and \cite{kohler2019linear}. The considered sets $\mathbb{W}=\{w\in\mathbb{R}^n |\; H_ww\leq h_w\}$ and $\Theta_0=\{\theta\in\mathbb{R}^p |\; H_{\theta_0}\theta\leq h_{\theta_{0}}\}$ are assumed to be polytopic, with $H_w\in\mathbb{R}^{n_w\times n}$, $h_w\in\mathbb{R}^{n_w}$, $H_{\theta_0}\in\mathbb{R}^{n_\theta \times p}$ and $h_{\theta_0}\in\mathbb{R}^{n_\theta}$. In order to ensure polytopic sets $\Theta_k$, polytopic set-membership estimation is used, which consists of computing the set of all possible parameters that explain the system evolution given a set of possible disturbances $\mathbb{W}$. 
For the considered dynamics \eqref{eq:lindynamics}, given state measurements $x_{k-1}$ and $x_k$, this non-falsified set of parameters is given by
\begin{equation}
\Delta_k =\{\theta\in\mathbb{R}^p |\; x_k-(A(\theta)x_{k-1}+B(\theta)u_{k-1})\in\mathbb{W}\},
\end{equation}
which is polytopic and whose explicit formulation is given in \cite{lorenzen2019robust}.
The parameter set $\Theta_k$ is updated by taking the intersection of the previous set $\Theta_{k-1}$ and the non-falsified parameter set
$\Theta_k=\Theta_{k-1}\cap\Delta_k.$

A major drawback of the proposed identification scheme is the potential increase in complexity of the resulting parameter sets through the addition of new half-spaces at every set update, which increases the computational complexity of the proposed adaptive MPSC scheme.
This issue can be addressed by fixing the shape of the parameter polytopes, e.g., by fixing the half-spaces, i.e., $\Theta_k=\{\theta\in\mathbb{R}^p|\;H_{\theta}\theta\leq h_{\theta_k}\}$, and only recomputing the right-hand side of the polytope inequality $h_{\theta_k}$ through the solution of a linear program (LP), as is shown in \cite{lorenzen2019robust}.
To further increase the computational update efficiency of the parameter identification as well as the efficiency of the proposed adaptive MPSC scheme, the set of parameters can be restricted to hypercubes with centre $\bar{\theta}_k\in\mathbb{R}^p$ and size $\eta_k\geq0$, i.e., $\Theta_k=\{\bar{\theta}_k\}\oplus\eta_k\mathbb{B}_p$ as described in \cite{kohler2019linear} and \cite{kohler2020nonlinear}.
This parametrisation results in $2p$ LPs to find the minimal and maximal values of $\theta$ in $\Theta_{k-1}\cap\Delta_k$ in every parameter dimension, thereby computing the smallest bounding hypercube of the intersection.

By using a tube controller $\kappa(x,v_{l|k})=Kx+v_{l|k}$ and a homothetic tube formulation for the sets
$\mathbb{X}_{l|k}=\{z_{l|k}\}\oplus\alpha_{l|k}\mathbb{X}_0$, with $\mathbb{X}_0=\{x\in\mathbb{R}^n|H_xx\leq\mathbf{1}\}$, $H_x\in\mathbb{R}^{n_x\times n}$ and $\alpha_{l|k}\geq0$, the optimisation problem \eqref{eq:optimisation} can be formulated as a quadratic program if $\mathbb{S}_f$ is also a polytope, with optimisation variables $v_{l|k}, z_{l|k}$ and $\alpha_{l|k}$. In \cite{kohler2019linear}, $z_{l|k}$ are computed according to dynamics \eqref{eq:lindynamics} with the center of the hypercube $\Theta_k$ as parameters, allowing for a more computationally efficient reformulation. 
Note that in the homothetic tube formulations, the terminal constraint \eqref{eq:terminalconstr} is given by $(z_{N|k},\alpha_{N|k})\in\mathbb{X}_f$, where the terminal set $\mathbb{X}_f$ is a set of translations and dilations $(z,\alpha)$. which can be iteratively enlarged through previously solved instances of \eqref{eq:optimisation} as shown in the next section.

\subsection{Iterative Terminal Set Enlargement for Recursive Feasibility}\label{sec:ittermrecfeas}
A recursively feasible MPSC problem facilitates the implementation of Algorithm~1, however it introduces the design task of finding a possibly large terminal set in order to reduce conservative safety interventions. We thus propose a mechanism to iteratively enlarge a terminal set for the homothetic tube formulation as described in Section~\ref{sec:efficientAMPSC} such that recursive feasibility is guaranteed when employing this new terminal set.
As discussed in Section~\ref{sec:recursiveMPSC}, we select the terminal set $\mathbb{X}_f$ such that for every translation and dilation $(z,\alpha)\in\mathbb{X}_f$, a translation and dilation in the terminal set exists at the next time step, ensuring recursive feasibility of \eqref{eq:optimisation}. This assumption on the terminal set is common in the robust adaptive MPC literature and is stated explicitly in \cite{lorenzen2019robust}, \cite{kohler2019linear} and \cite{kohler2020nonlinear} and used implicitly in \cite{lu2019robust} in the online optimisation problem.
\begin{assumpt} \label{ass:homothetictubeterm}
Let $\mathbb{W}$, $\bar{\Theta}$ be polytopic and let $\mathbb{X}_{l|k}$ be of the form $\{z_{l|k}\}\oplus\alpha_{l|k}\mathbb{X}_0$ with polytopic $\mathbb{X}_0=\{H_xx\leq\mathbf{1}\}$ and $H_x\in\mathbb{R}^{n_x\times n}$.
Consider a non-empty terminal set $\mathbb{X}_f=\{(z,\alpha)| H_Tz+h_T\alpha\leq\mathbf{1}\}$, with $H_T\in\mathbb{R}^{n_T\times n}$ and $h_T\in\mathbb{R}^{n_T}$, and a tube control law $\kappa(x,v)=Kx+v$ in \eqref{eq:optimisation}. For every $(z,\alpha)\in\mathbb{X}_f$ , there exists a $v$ and $(z^+,\alpha^+)\in\mathbb{X}_f$, such that for all $\theta\in\bar{\Theta}$ and $x\in\{z\}\oplus\alpha\mathbb{X}_0$, it holds that $(x,\kappa(x,v))\in\mathbb{Z}$ and $A(\theta)(\{z\}\oplus\alpha\mathbb{X}_0)\oplus B(\theta)K(\{z\}\oplus\alpha\mathbb{X}_0)\oplus\{ B(\theta)v\}\oplus\mathbb{W}\subseteq\{z^+\}\oplus\alpha^+\mathbb{X}_0.$
\end{assumpt}

By using solved instances of \eqref{eq:optimisation} with optimal $(z_{l|k}^*,\alpha_{l|k}^*)$, the terminal set can then be enlarged, such that a recursively feasible optimisation problem is recovered. 

\begin{thrm}
 Let Assumptions~\ref{ass:bounded},\ref{ass:parameter} and \ref{ass:homothetictubeterm} hold and \eqref{eq:optimisation} be convex, then the set
\begin{equation}\label{eq:homotheticsetenlargement}
	\mathbb{X}_f^{\mathcal{M}(\bar{k})}=\textup{co}\left(\{(z^*_{l|k},\alpha^*_{l|k})\}_{l\in\mathcal{I}_{[0,N]},k\in\mathcal{M}(\bar{k})},\mathbb{X}_f\right)
\end{equation}
satisfies Assumption~\ref{ass:homothetictubeterm} with respect to $\bar{\Theta}=\Theta_{\bar{k}}$.
\end{thrm}

\begin{proof}
We denote the $n_X$ vertices of $\mathbb{X}_f$ as $\{(z^*_{0|k},\alpha^*_{0|k})\}_{k\in\mathcal{I}_ {[-n_X,-1]}}$ and construct corresponding tuples $(z^*_{l|k},\alpha^*_{l|k}){\in}\mathbb{X}_f$ for $l{\in}\mathcal{I}_{[1,N+1]}$ that satisfy Assumption~\ref{ass:homothetictubeterm} for consecutive pairs $l$ and $l{+}1$. We then denote as $\mathcal{N}{=}\mathcal{M}(\bar{k})\!\cup\mathcal{I}_{[-n_X,-1]}$ the set of solved time steps $\mathcal{M}(\bar{k})$, together with the $n_X$ constructed solutions for each vertex of $\mathbb{X}_f$. Through the solutions of \eqref{eq:optimisation} and the constructed solutions we have for all $l\in\mathcal{I}_{[0,N]}$, $k\in\mathcal{N}$ and $\theta\in\Theta_{\bar{k}}\subseteq\Theta_k$ , that it holds that
\vspace{-0.1cm}
\begin{equation*}
\vspace{-0.1cm}
A_{\!cl}(\theta)(\!\{z_{l|k}^*\}{\oplus}\alpha^*_{l|k}\mathbb{X}_0\!){\oplus}\{\!B(\theta)v_{l|k}^*\!\}{\oplus}\mathbb{W}{\subseteq}\{\!z_{l+1|k}^*\!\}{\oplus}\alpha_{l+1|k}^*\mathbb{X}_0,
\end{equation*}
where we define $A_{cl}(\theta){=}A(\theta){+}B(\theta)K$ and use the fact that $(z^*_{N+1|k},\alpha^*_{N+1|k}){\in}\mathbb{X}_f$ exists according to Assumption~\ref{ass:homothetictubeterm} as $(z^*_{N|k},\alpha^*_{N|k}){\in}\mathbb{X}_f$. 
For any $(z,\alpha)\in\mathbb{X}_f^{\mathcal{M}(\bar{k})}$, we can write $(z,\alpha)=\sum_{l\in\mathcal{I}_{[0,N]}}\sum_{k\in\mathcal{N}}\lambda_{l|k}(z_{l|k}^*,\alpha_{l|k}^*)$ due to the convex hull, where it holds that $\sum_{l\in\mathcal{I}_{[0,N]}}\sum_{k\in\mathcal{N}}\lambda_{l|k}{=}1$, $\lambda_{l|k}{\geq}0$.
We then choose $v{=}\sum_{l\in\mathcal{I}_{[0,N]}}\sum_{k\in\mathcal{N}}\lambda_{l|k}v^*_{l|k}$ where $v^*_{l|k}$ corresponds to the input solution of \eqref{eq:optimisation} at time step $l|k$, and the corresponding $(z^+,\alpha^+){=}\sum_{l\in\mathcal{I}_{[0,N]}}\sum_{k\in\mathcal{N}}\lambda_{l|k}(z^*_{l+1|k},\alpha^*_{l+1|k})$. 
It then follows that for all $\theta\in\Theta_{\bar{k}}$,
\begin{gather}
\begin{align}
&A_{cl}(\theta)\left(\{z\}\oplus\alpha\mathbb{X}_0\right)\oplus\{B(\theta)v\} \oplus\mathbb{W}\nonumber\\
=&A_{cl}(\theta)\Big(\{{\sum_{l\in\mathcal{I}_{[0,N]}}}{\sum_{k\in\mathcal{N}}}\lambda_{l|k}z_{l|k}^*\}
{\oplus}\big({\sum_{l\in\mathcal{I}_{[0,N]}}}{\sum_{k\in\mathcal{N}}}\lambda_{l|k}\alpha_{l|k}^*\big)\mathbb{X}_0\Big) \nonumber\\
&{\oplus}\{B(\theta){\sum_{l\in\mathcal{I}_{[0,N]}}}{\sum_{k\in\mathcal{N}}}\lambda_{l|k}v_{l|k}^*\}
{\oplus}\big({\sum_{l\in\mathcal{I}_{[0,N]}}}{\sum_{k\in\mathcal{N}}}\lambda_{l|k}\big)\mathbb{W}\nonumber\\
=&{\bigoplus_{l\in\mathcal{I}_{[0,N]}}}{\bigoplus_{k\in\mathcal{N}}}\lambda_{l|k}\big(A_{cl}(\theta)(\{z_{l|k}^*\}{\oplus}\alpha_{l|k}^*\mathbb{X}_0){\oplus}\{B(\theta)v_{l|k}^*\}{\oplus}\mathbb{W}\big) \nonumber \\
\subseteq&{\bigoplus_{l\in\mathcal{I}_{[0,N]}}}{\bigoplus_{k\in\mathcal{N}}}{\lambda_{l|k}}(\{z_{l+1|k}^*\}\oplus\alpha_{l+1|k}^*\mathbb{X}_0)\nonumber\\
=&\{z^+\}\oplus\alpha^+\mathbb{X}_0,\nonumber
\end{align}
\end{gather}
where step 2 is shown in detail for convex sets $\mathbb{A}$ and $\mathbb{B}$: \\ 
$(\sum_i\lambda_i)\mathbb{A}\oplus(\sum_i\lambda_i)\mathbb{B}=\{\sum_i\lambda_ia+\sum_i\lambda_ib|a\in\mathbb{A}, b\in\mathbb{B}\} \\
=\{\sum_i\tilde{a}_i+\sum_i\tilde{b}_i|\tilde{a}_i\in\lambda_i\mathbb{A},\tilde{b}_i\in\lambda_i\mathbb{B}\}=\bigoplus_i(\lambda_i\mathbb{A}\oplus\lambda_i\mathbb{B})\\
=\bigoplus_i\lambda_i(\mathbb{A}\oplus\mathbb{B})$.
As $(z,a)\in\mathbb{X}_f^{\mathcal{M}(\bar{k})}$, it follows that the tuple $(z^+,\alpha^+)\in\mathbb{X}_f^{\mathcal{M}(\bar{k})}$ from a convex combination of the tuples $(z^*_{l+1|k},\alpha^*_{l+1|k}){\in}\mathbb{X}_f^{\mathcal{M}(\bar{k})}$. Similarly through the convex combination of $v$, the combined state and input constraints are guaranteed to hold.
\end{proof}

\begin{remark}
Given a representation of the set $\mathbb{X}_0{=}\textup{co}(x^1, x^2, \dots, x^{n_{X_0}})$ with $n_{X_0}$ vertices, the terminal set enlargement in \eqref{eq:homotheticsetenlargement} can be further improved with the vertices of the previously computed $\mathbb{X}^*_{l|k}$ by using $\textup{co}\left(\mathbb{X}_f^{\mathcal{M}(\bar{k})}, \{(z^*_{l|k}+\alpha^*_{l|k}x^j, 0)\}_{j\in\mathcal{I}_{[0,n_{X_0}]}, l\in\mathcal{I}_{[0,N]},k\in\mathcal{M}(\bar{k})}\right)$ as for all $j{\in}\mathcal{I}_{[0,n_{X_0}]}$, $z^*_{l|k}{+}\alpha^*_{l|k}x^j{\in}\mathbb{X}^*_{l|k}$, which implies that $\forall\theta\in\Theta_{\bar{k}}$, $A_{cl}(\theta)(z^*_{l|k}{+}\alpha^*_{l|k}x^j){\oplus}\{B(\theta)v^*_{l|k}\}{\oplus}\mathbb{W}{\subseteq}\mathbb{X}^*_{l+1|k}$.
\end{remark}
\section{NUMERICAL EXAMPLE}\label{sec:numericalexample}
We consider a chain of $n_{\textup{MSD}}$ mass elements connected by $n_{\textup{MSD}}-1$ springs and dampers. The discrete-time dynamics of the mass element $i$ are given by
\begin{gather}
\begin{align}
p_{k+1,i}&{=}p_{k,i}+T_sv_{k,i} \nonumber\\
v_{k+1,i}&{=}v_{k,i}{-}T_sc_{i-1,i}(p_{k,i}{-}p_{k,i-1}){-}T_sd_{i-1,i}(v_{k,i}{-}v_{k,i-1}) \nonumber \\
&{+}T_sc_{i,i+1}(p_{k,i+1}{-}p_{k,i}){-}T_sd_{i,i+1}(v_{k,i+1}{-}v_{k,i}){+}u_{k,i} \nonumber
\end{align}
\end{gather}
with the position of element $i$ at time step $k$ denoted by $p_{k,i}$ and the element velocity $v_{k,i}$, sampling time $T_s=0.2s$, spring and damping constants $c_{i,i+1}$ and $d_{i,i+1}$ of the springs and dampers connecting elements $i$ and $i+1$.
All damping coefficients $d_{i,i+1}=0.1$, with $d_{0,1}$ and $d_{n_{\textup{MSD}},n_{\textup{MSD}}+1}$ and the corresponding spring constants being 0. The remaining spring constants are randomly drawn between $[0.05, 0.25]$ and are considered as uncertain parameters $\theta$ with the initial set of parameters $\Theta_0=[0.05, 0.25]^{n_{\textup{MSD}}-1}$, such that $\theta^*\in\Theta_0$ and an additive disturbance on the positions and velocities with $\abs{w}\leq1\textup{e}-3$ is used. The dynamics can thus be defined as $x_{k+1}=A(\theta)x_k+Bu_k+w_k$ and Assumption~\ref{ass:bounded} is fulfilled. 

\begin{figure*}[thpb]
\vspace{-0.2cm}
   \hspace{-2cm}\includegraphics[width=1.2\textwidth]{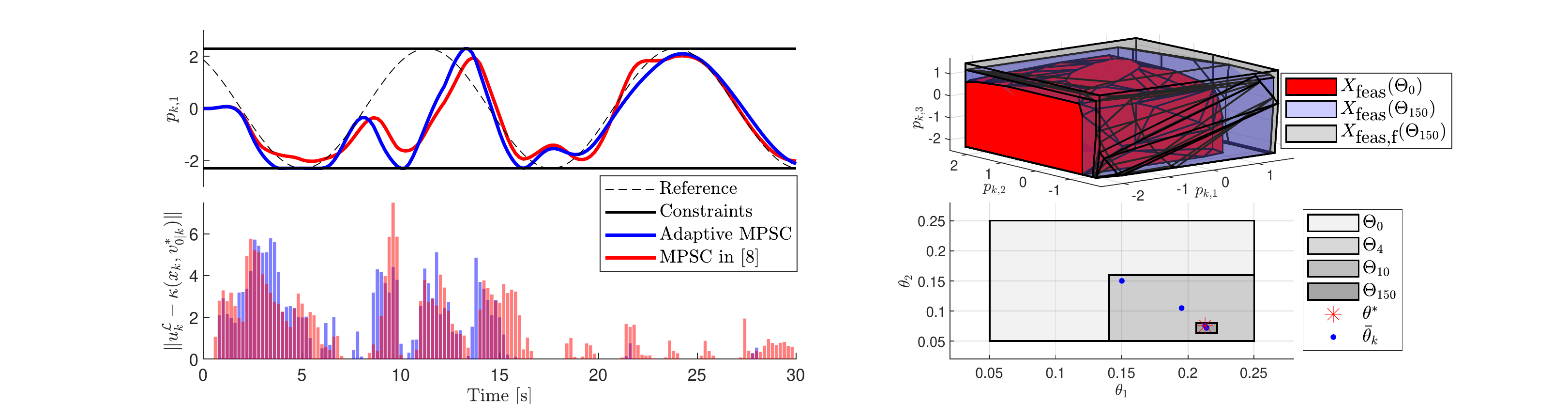}
\vspace{-0.7cm}
   \caption{Simulation of 3 mass-spring-damper elements using the adaptive MPSC scheme. The approach is compared against \cite{wabersich2018linear} (left) demonstrating less frequent safety filter interventions while ensuring constraint satisfaction and successfully identifying the unknown parameters (bottom right). Additionally, the initial and final feasible set as well as the feasible set after terminal set enlargement $\mathbb{X}_{\textup{feas}}(\Theta_0)$, $\mathbb{X}_{\textup{feas}}(\Theta_{150})$ and $\mathbb{X}_{\textup{feas,f}}(\Theta_{150})$, respectively, are shown for the element positions for a fixed velocity (top right).}
\vspace{-0.2cm}
   \label{fig:sim}
\end{figure*}
%
The system is simulated for $30$s with $3$ and $8$ elements from the origin using the adaptive MPSC scheme with the constraint reformulation in \cite{kohler2019linear} with decoupled state and input constraints $\mathbb{X}=[-2.3,2.3]^{2n_{\textup{MSD}}}$ and $\mathbb{U}=[-3.5,3.5]^{n_{\textup{MSD}}}$. The problem is solved using YALMIP \cite{lofberg2004yalmip} and MOSEK \cite{anderson2000mosek} with an average computation time of $6$ms and $340$ms for 3 and 8 elements respectively. A PGSD controller \cite{kolter2009PGSD} is used with random initial control parameters which are trained during the simulation. 
The terminal constraints used are $z_{N|k}=0$ and $0\leq\alpha_{N|k}\leq1$. The simulation results with 3 elements and a comparison to the method in \cite{wabersich2018linear} where the parametric uncertainty is included by enlarging the disturbance set can be seen in Figure \ref{fig:sim}. The adaptive MPSC scheme successfully prevents constraint violations of the system and interferes less conservatively than \cite{wabersich2018linear}. The enlargement of the implicitly defined safe set of the MPSC through parameter adaptation and through an additional terminal set enlargement using all available data after the simulation is finished is also shown in Figure \ref{fig:sim}. A section of the feasible sets $\mathbb{X}_{\textup{feas}}(\Theta_0)$, $\mathbb{X}_{\textup{feas}}(\Theta_{150})$ and with the enlarged terminal set, $\mathbb{X}_{\textup{feas,f}}(\Theta_{150})$, for a fixed velocity is shown, which is computed through gridding of the state space. In order to compute the volume of the implicitly defined safe sets in $2n_{\textup{MSD}}$ dimensions, Monte Carlo Integration is used with $10^5$ randomly drawn samples. The results and a comparison with the feasible set $\mathbb{X}_{\textup{feas,\cite{wabersich2018linear}}}$ of \cite{wabersich2018linear} and the total volume within the constraints are shown in Table 1, where a $21\%$ increase in volume is observed for the case of $3$ mass-spring-damper elements after the parameter estimation and a total increase of $28\%$ with the terminal set enlargement. For $8$ elements, an increase of $100\%$ and $120\%$, respectively, is observed.

\section{CONCLUSION}
An adaptive Model Predictive Safety Certification scheme was proposed, which ensures safety of dynamical systems controlled by any learning-based controller. This modular framework uses set-membership estimation in order to sequentially improve the set in which uncertain parameters can possibly lie. The parameter estimation allows to enlarge the feasible set of the MPSC, and thereby the safe set of operation, in an online manner with recursive feasibility guarantees. We provide a possible enlargement of the terminal safe set used in the MPSC optimisation problem using previously solved instances, in order to further increase the feasible set of the MPSC and present a design method allowing for a computationally efficient optimisation problem. The adaptive MPSC scheme was applied to a chain of mass-spring-damper elements, which showed a significant increase in the implicit safe set volume through the parameter estimation and the terminal safe set enlargement and interfered less often than the nominal method in \cite{wabersich2018linear}.

\begin{table}
\vspace{0.22cm}
\setlength{\tabcolsep}{4.7pt}
\caption{Volume of the feasible set of the adaptive MPSC Optimisation Problem through Monte Carlo Integration}
\label{table:ex}
\vspace{-0.4cm}
\begin{center}
\begin{tabular}{|c|c|c|c|c|c|}
\hline
 &\!$\mathbb{X}_{\textup{feas}}(\Theta_0)$\!&\!$\mathbb{X}_{\textup{feas}}(\Theta_{150})$\!&\!$\mathbb{X}_{\textup{feas,f}}(\Theta_{150})$\!&\!$\mathbb{X}_{\textup{feas,\cite{wabersich2018linear}}}$\!&\!Constraint\!\\
\!\#\! MSD\!& Volume & Volume & Volume & Volume & Volume\\
\hline
$3$ & $5.86\textup{e}3$ & $7.10\textup{e}3$ & $7.46\textup{e}3$ & $3.62\textup{e}3$ & $9.47\textup{e}3$\\
\hline
$8$ & $9.64\textup{e}9$ & $1.93\textup{e}10$ & 2.12\textup{e}10 & $\diagup$ \footnotemark & $4.02\textup{e}10$\\
\hline
\end{tabular}
\end{center}
\end{table}	

\footnotetext{The RPI set for \cite{wabersich2018linear} could not be computed due to the complexity of the disturbance set resulting from the use of nominal linear dynamics matrices.}

\addtolength{\textheight}{-11cm}   
\newpage

\begin{appendix}

\subsection{Nonlinear Extension}\label{sec:nonlinearextension}
The adaptive MPSC scheme provided can be extended to nonlinear dynamics 
\begin{equation}\label{eq:NLdynamics}
	x_{k+1}=f(x_k,u_k,w_k,\theta)
\end{equation}
subject to a compact nonlinear constraint set $(x_k,u_k)\in\mathbb{Z}=\{(x,u)\in\mathbb{R}^n\times\mathbb{R}^m|\; H(x,u)\leq\mathbf{1}\}$, similar to \cite[Section~2]{kohler2020nonlinear}.
In order to robustly provide safety guarantees for the adaptive MPSC scheme, we make the following assumption.
\begin{assumpt}\label{ass:NLmap}
There exists a map $\Phi:2^{\mathbb{R}^n}\times\mathbb{R}^m\times2^{\mathbb{R}^p}\rightarrow 2^{\mathbb{R}^n}$, such that for any $(\mathcal{X},u)\subseteq\mathbb{Z}, \Theta\subseteq\Theta_0$, we have
\begin{equation}
f(x,u,w,\theta)\in\Phi(\mathcal{X},u,\Theta)
\end{equation}
for any $x\in\mathcal{X}, w\in\mathbb{W}$ and $\theta\in\Theta$. Furthermore, the map $\Phi$ satisfies the following monotonicity property
\begin{equation}
	\Phi(\mathcal{X}',u,\Theta')\subseteq\Phi(\mathcal{X},u,\Theta)
\end{equation}
for any $(\mathcal{X}',u)\subseteq(\mathcal{X},u)\subseteq\mathbb{Z}$ and any $\Theta'\subseteq\Theta\subseteq\Theta_0$.
\end{assumpt}
The adaptive MPSC scheme then consists of Algorithm~1 where the constraint \eqref{eq:tubeincconstr} is replaced by
\begin{equation} \label{eq:optimisationNL}
\Phi(\mathbb{X}_{l|k}, \kappa(x,v_{l|k}), \Theta_k) \subseteq \mathbb{X}_{l+1|k}\; \forall x\in\mathbb{X}_{l|k}.
\end{equation}

Constraint \eqref{eq:optimisationNL} ensures that the states predicted with respect to the inputs $\kappa(x,v_{l|k})$ are contained within the state tube constructed from the sets $\mathbb{X}_{l|k}$ through the use of the map $\Phi$ from Assumption~\ref{ass:NLmap}. Similarly to the linear case, constraint satisfaction is ensured through constraint \eqref{eq:stateconstr} and thus the solution of the nonlinear optimisation problem $\kappa(x,v^*_{l|k})$ is guaranteed to lead the system to the terminal safe set $\mathbb{S}_f$. The monotonicity property in Assumption~\ref{ass:NLmap} guarantees that through a parameter set update $\Theta_{k+1}\subseteq\Theta_{k}$, the inputs $\kappa(x,v^*_{l|k})$ still ensure constraint satisfaction as it holds that $\Phi(\mathbb{X}_{l|k},\kappa(x,v^*_{l|k}),\Theta_{k+1})\subseteq\Phi(\mathbb{X}_{l|k},\kappa(x,v^*_{l|k}),\Theta_{k})\subseteq\mathbb{X}^*_{l+1|k}$.
\begin{thrm}
Let Assumptions~\ref{ass:parameter} and \ref{ass:NLmap} hold. The control law $\pi_{\textup{MPSC}}(u_{\mathscr{L}},x_k,\Theta_k,k)$ resulting from Algorithm~1 with the constraint \eqref{eq:optimisationNL} is a safe backup controller and $\mathbb{X}_{\textup{feas}}(\Theta_k)\cup\mathbb{S}_f$ the corresponding safe set according to Definition~\ref{def:safeset}. Additionally, it holds that $\mathbb{X}_{\textup{feas}}(\Theta_{0})\cup\mathbb{S}_f\subseteq\mathbb{X}_{\textup{feas}}(\Theta_{k-1})\cup\mathbb{S}_f\subseteq\mathbb{X}_{\textup{feas}}(\Theta_{k})\cup\mathbb{S}_f$ for all time steps $k>0$.
\end{thrm}
\begin{proof}
The proof of this theorem follows the proof of Theorem~\ref{thrm:AMPSC} with the constraint \eqref{eq:optimisationNL}.
\end{proof}

For dynamics which depend affinely on the uncertain parameters $x_{k+1}=g(x_k,u_k)+h(x_k,u_k)\theta$, the nonlinear constraints can be simplified under further assumptions as shown in \cite{kohler2020nonlinear} and \cite{goncalves2016robust}.
In \cite{kohler2020nonlinear} the explicit computation of an incremental exponential Lyapunov function
is used for recursive feasibility. However, this introduces a difficult design problem, which can be relaxed in our case through \eqref{eq:terminalconstr} in combination with Definition~\ref{def:safeset}, such that only an incremental Lyapunov function subject to \cite[Assumption~5]{kohler2020nonlinear} for the tube design is required as well as a terminal safe set. 

The proposed nonlinear adaptive MPSC scheme provides a deterministic safety guarantee by using the a priori known set $\Theta_0$, as opposed to \cite{wabersich2021safe}, where distributions with potentially unbounded $\theta$ are considered, resulting in chance constraint satisfaction. 
\end{appendix}

\end{document}